\documentclass{llncs}
\usepackage[english]{babel}
\usepackage{graphicx}
\usepackage[cmex10]{amsmath}
\usepackage{amssymb}
\usepackage[utf8]{inputenc}
\usepackage[breaklinks,bookmarks=false]{hyperref}

\newcommand{\kmismatches}{\overset{k}{\rule{0pt}{.4ex}\smash{\sim}}}
\newcommand{\mismatches}[1]{\overset{#1}{\sim}}
\newcommand{\bigo}{\mathcal{O}}
\newcommand{\pred}{\mathrm{pred}}

\newcommand{\ham}{\mathrm{H}}

\newcommand{\polylogws}{\text{\hspace{1mm}polylog}}
\title{Order-preserving pattern matching\\ with $k$ mismatches}

\begin{document}
\author{Paweł Gawrychowski\inst{1} and Przemysław Uznański\inst{2}\thanks{This work was started while the second author was a PhD student at Inria Bordeaux Sud-Ouest, France.}}
\institute{
Max-Planck-Institut f\"{u}r Informatik, Saarbr\"ucken, Germany
\and
LIF, CNRS and Aix-Marseille Universit\'e, Marseille, France}

\maketitle

\begin{abstract}
We study a generalization of the order-preserving pattern matching recently introduced by Kubica~et~al. (Inf. Process. Let., 2013) and Kim~et~al. (submitted to Theor. Comp. Sci.), where instead of looking for an exact copy of the pattern, we only require that the relative order between the elements is the same. In our variant, we additionally allow up to $k$ mismatches between the pattern of length $m$ and the text of length $n$, and the goal is to construct an efficient
algorithm for small values of $k$. Our solution detects an order-preserving occurrence with up to $k$ mismatches in $\bigo(n(\log\log m+k\log\log k))$ time.
\end{abstract}


\section{Introduction}

\emph{Order-preserving pattern matching}, recently introduced in~\cite{Kim} and~\cite{KubicaOrder}, and further considered in~\cite{CrochemoreOrder}, is a variant of the well-known pattern matching problem, where instead of looking for a fragment of the text which is identical to the given pattern, we are interested in locating a fragment which is order-isomorphic with the pattern. Two sequences over integer alphabet are \emph{order-isomorphic} if the relative order between any two elements at the same positions in both sequences is the same. Similar problems have been extensively studied in a slightly different setting, where instead of a fragment, we are interested in a (not necessarily contiguous) subsequence. For instance, pattern avoidance in permutations was of much interest.

For the order-preserving pattern matching, both~\cite{Kim} and~\cite{KubicaOrder} present an $\bigo(n+m\log m)$ time algorithm, where $n$ is the length of the text, and $m$ is the length of the pattern. Actually, the solution given by~\cite{KubicaOrder} works in $\bigo(n+\text{sort}(m))$ time, where $\text{sort}(m)$ is the time required to sort a sequence of $m$ numbers. Furthermore, efficient algorithms for the version with multiple patterns are known~\cite{CrochemoreOrder}. Also, a generalization of suffix trees in the order-preserving setting was recently considered~\cite{CrochemoreOrder}, and the question of constructing a forward automaton allowing efficient pattern matching and developing an average-case optimal pattern matching algorithm was studied~\cite{Vialette}.

Given that the complexity of the exact order-preserving pattern matching seems to be already settled, a natural direction is to consider its approximate version. Such direction was successfully investigated for the related case of \emph{parametrized pattern matching} in~\cite{MosheParam}, where an $\bigo(nk^{1.5}+mk\log m)$ time algorithm was given for parametrized matching with $k$ mismatches.

We consider a relaxation of order-preserving pattern matching, which we call \emph{order-preserving pattern matching with $k$ mismatches}. Instead of requiring that the fragment we seek is order-isomorphic with the pattern, we are allowed to first remove $k$ elements at the corresponding positions from the fragment and the pattern, and then check if the remaining two sequences are order-isomorphic. In such setting, it is relatively straightforward to achieve running time of $\bigo(nm\log\log m)$, where $n$ is the length of the text, and $m$ is the length of the pattern. Such complexity might be unacceptable for long patterns, though, and we aim to achieve complexity of the form $\bigo(n f(k))$. In other words, we would like our running time to be close to linear if the bound on the number of mismatches is very small. We construct a deterministic algorithm with $\bigo(n(\log\log m+k\log\log k))$ running time. At a very high level, our solution is similar to the one given in~\cite{MosheParam}. We show how to filter the possible starting positions so that a position is either eliminated in $\bigo(f(k))$ time, or the structure of the fragment starting there is simple, and we can verify the occurrence in $\bigo(f(k))$ time. The details are quite different in our setting, though.

A different variant of approximate order-preserving pattern matching could be that we allow to remove $k$ elements from the fragment, and $k$ elements from the pattern, but don't require that they are at the same positions. Then we get order-preserving pattern matching with $k$ errors. Unfortunately, such modification seems difficult to solve in polynomial time: even if the only allowed operation is removing $k$ elements from the fragment, the problem becomes NP-complete~\cite{Bose}.

\section{Overview of the algorithm}
Given a text $(t_{1},\ldots,t_{n})$ and a pattern $(p_{1},\ldots,p_{m})$, we want to locate an order-preserving occurrence with at most $k$ mismatches of the pattern in the text. Such occurrence is a fragment $(t_{i},\ldots,t_{i+m-1})$ with the property that if we ignore the elements at some up to corresponding $k$ positions in the fragment and the pattern, the relative order of the remaining elements is the same in both of them.

The above definition of the problem is not very convenient to work with, hence we start with characterising $k$-isomorphic sequences using the language of subsequences in Lemma~\ref{lemma:subseq}. This will be useful in some of the further proofs and also gives us a polynomial time solution for the problem, which simply considers every possible $i$ separately. To improve on this naive solution, we need a way of quickly eliminating some of these starting positions. For this we define the signature $S(a_{1},\ldots,a_{m})$ of a sequence $(a_{1},\ldots,a_{m})$, and show in Lemma~\ref{lemma:implication} that the Hamming distance between the signatures of two $k$-isomorphic sequences cannot be too large. Hence such distance between $S(t_{i},\ldots,t_{i+m-1})$ and $S(p_{1},\ldots,p_{m})$ can be used to filter some starting positions where a match cannot happen.

In order to make the filtering efficient, we need to maintain $S(t_{i},\ldots,t_{i+m-1})$ as we increase $i$, i.e., move a window of length $m$ through the text. For this we first provide in Lemma~\ref{lemma:prune_structure} a data structure which, for a fixed word, allows us to maintain a word of a similar length under changing the letters, so that we can quickly generate the first $k$ mismatches between subwords of the current and the fixed word. The structure is based on representing the current word as a concatenation of subwords of the fixed word. Then we observe that increasing $i$ by one changes the current signature only slightly, which allows us to apply the aforementioned structure to maintain $S(t_{i},\ldots,t_{i+m-1})$ as shown in Lemma~\ref{lemma:prune}. Therefore we can efficiently eliminate all starting positions for which the Hamming distance between the signatures is too large.

For all the remaining starting positions, we reduce the problem to computing the heaviest increasing subsequence, which is a weighted version of the well-known longest increasing subsequence, in Lemma~\ref{lemma:reduction}. The time taken by the reduction depends on the Hamming distance, which is small as otherwise the position would be eliminated in the previous step. Finally, such weighted version of the longest increasing subsequence can be solved efficiently as shown in Lemma~\ref{lemma:heaviest}. Altogether these results give an algorithm for order-preserving pattern matching with $k$
with the cost of processing a single $i$ depending mostly on $k$.

An implicit assumption in this solution is that there are no repeated elements in neither the text nor the pattern. In the last part of the paper we remove this assumption while keeping the same time complexity. At a high level the algorithm remains the same, but a few carefully chosen modifications are necessary. First we further generalize the heaviest increasing subsequence into heaviest chain in a plane, and in Lemma~\ref{lemma:heaviest chain} how to solve this version efficiently. Then we modify the definition of a signature, and prove in Lemma~\ref{lemma:reduction2} that after such change checking if two sequences are order-isomorphic can be reduced to computing the heaviest chain. 

\section{Preliminaries}

We consider strings over an integer alphabet, or in other words sequences of integers. Two such sequences are \emph{order-isomorphic} (or simply \emph{isomorphic}), denoted by $(a_{1},\ldots,a_{m})\sim (b_{1},\ldots,b_{m})$, when $a_{i} \leq a_{j}$ iff $b_{i} \leq b_{j}$ for all $i,j$. 
We will also use the usual equality of strings. Whenever we are talking about sequences, we are interested in the relative order between their elements, and whenever we are talking about strings consisting of characters, the equality of elements is of interest to us. For two strings $s$ and $t$, their \emph{Hamming distance} $\ham(s,t)$ is simply the number of positions where the corresponding characters differ.

Given a text $(t_{1},\ldots,t_{n})$ and a pattern $(p_{1},\ldots,p_{m})$, the \emph{order-preserving pattern matching} problem is to find $i$ such that
$(t_{i},\ldots,t_{i+m-1}) \sim (p_{1},\ldots,p_{m})$.  We consider its approximate version, i.e., order-preserving pattern matching with $k$ mismatches. 

We define two sequences \emph{order-isomorphic with $k$ mismatches}, denoted by 
$(a_{1},\ldots,a_{m})\kmismatches (b_{1},\ldots,b_{m})$, when we can select (up to) $k$ indices $1\leq i_{1} < \ldots < i_{k}\leq m$, and remove the corresponding elements from both sequences so that the resulting two new sequences are isomorphic, i.e., $a_{j} \leq a_{j'}$ iff $b_{j} \leq b_{j'}$ for any $j,j' \notin\{i_{1},\ldots,i_{k}\}$. In \emph{order-preserving pattern matching with $k$ mismatches} we want $i$ such that $(t_{i},\ldots,t_{i+m-1}) \kmismatches (p_{1},\ldots,p_{m})$, see Fig.~\ref{fig:occurrence}.

\begin{figure}[t]
\begin{center}
\includegraphics[width=0.8\linewidth]{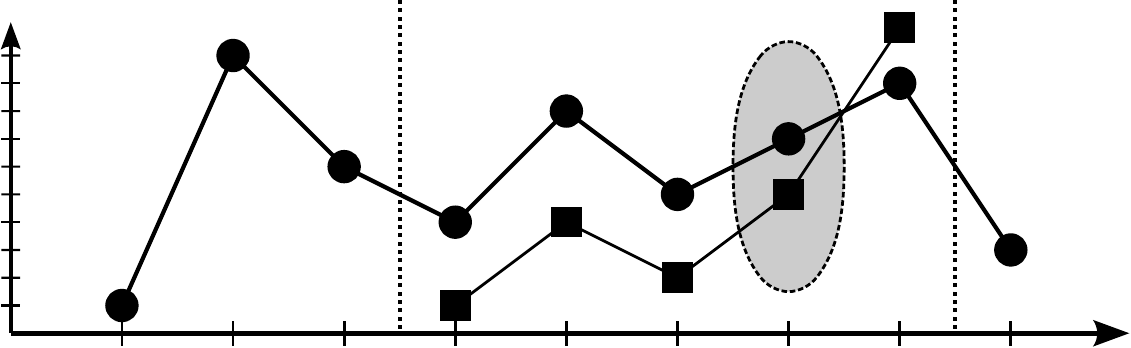}
\end{center}
\caption{$[1,4,2,5,11]$ occurs in $[1,10,6,4,8,5,7,9,3]$ (at position 4) with $1$ mismatch.}
\label{fig:occurrence}
\end{figure}

Our solution works in the word RAM model, where $n$ integers can be sorted in $\bigo(n\log\log n)$ time~\cite{Han}, and we can implement dynamic dictionaries using van Emde Boas trees. In the restricted comparison model, where we can only compare the integers, all $\log\log$ in our complexities increase to $\log$.

In Section~\ref{section:algorithm}, we assume that integers in any sequence are all distinct. Such assumption was already made in one of the papers introducing the problem~\cite{Kim}, with a justification that we can always perturb the input to ensure this (or, more precisely, we can consider pairs consisting of a number and its position). In some cases this can change the answer, though\footnote{More precisely, it might make two non-isomorphic sequences isomorphic, but not the other way around.}. Nevertheless, using a more complicated argument, shown in Section~\ref{section:generalization}, we can generalize our solution to allow the numbers to repeat. Another simplifying assumption that we make in designing our algorithm is that $n\leq 2m$. We can do so using a standard trick of cutting the text into overlapping fragments of length $2m$ and running the algorithm on each such fragment separately, which preserves all possible occurrences.


\section{The algorithm}
\label{section:algorithm}

First we translate $k$-isomorphism into the language of subsequences.

\begin{lemma}
\label{lemma:subseq}
$(a_{1},\ldots,a_{m}) \kmismatches (b_{1},\ldots,b_{m})$ iff there exist $i_1,\ldots,i_{m-k}$ such that $a_{i_1} < \ldots < a_{i_{m-k}}$ and $b_{i_1} < \ldots < b_{i_{m-k}}$.
\end{lemma}

\begin{proof}
\textcircled{$\Rightarrow$}
If the sequences are $k$-isomorphic, according to the definition there exists a set of $k$ indices $j_1,j_2,\ldots,j_k$ such that
$\forall_{j,j' \not\in \{j_1,\ldots,j_k\}} a_j < a_{j'} \text{ iff } b_j < b_{j'}$.
We set $\{i_1,\ldots,i_{m-k}\} = \{1,\ldots,m\} \setminus \{j_1,\ldots,j_k\}$ and order its elements so that
$a_{i_1} < \ldots < a_{i_{m-k}}$. Then clearly $b_{i_1} < \ldots < b_{i_{m-k}}$ as well.

\textcircled{$\Leftarrow$}
From the existence of $i_1,\ldots,i_{m-k}$ such that $a_{i_1} < \ldots < a_{i_{m-k}}$ and $b_{i_1} < \ldots < b_{i_{m-k}}$ we deduce that if we choose $\{j_1,\ldots,j_k\} = \{1,\ldots,m\} \setminus \{i_1,\ldots,i_{m-k}\}$, we have that $\forall_{j,j' \not\in \{j_1,\ldots,j_k\}} a_j < a_{j'} \text{ iff } b_j < b_{j'}$.
\qed
\end{proof}

The above lemma implies an inductive interpretation of $k$-isomorphism useful in further proofs and a fast method for testing $k$-isomorphism.

\begin{proposition}
\label{prop:inductive}
If $(a_{1},\ldots,a_{m}) \mismatches{k+1} (b_{1},\ldots,b_{m})$ then there exists $(a'_{1},\ldots,a'_{m})$ such that $(a_{1},\ldots,a_{m}) \mismatches{1} (a'_{1},\ldots,a'_{m})$ and $(a'_{1},\ldots,a'_{m}) \kmismatches (b_{1},\ldots,b_{m})$.
\end{proposition}

\begin{lemma}
\label{lemma:checking}
$(a_{1},\ldots,a_{m}) \kmismatches (b_{1},\ldots,b_{m})$ can be checked in time $\bigo(m \log\log m)$.
\end{lemma}

\begin{proof}
Let $\pi$ be the sorting permutation of $(a_{1},\ldots,a_{m})$. Such permutation can be found in time $\bigo(m \log\log m)$. Let $(b'_1,\ldots, b'_m)$ be a sequence defined 
by setting $b'_i := b_{\pi(i)}$. Then, by Lemma \ref{lemma:subseq}, $(a_1,\ldots,a_m) \kmismatches (b_1,\ldots,b_m)$ iff there exists an increasing subsequence of $b'$ of length $m-k$. Existence of such a subsequence can be checked in time $\bigo(m \log\log m)$ using a van Emde Boas tree~\cite{Hunt}.
\qed
\end{proof}

By applying the above lemma to each of the possible occurrences separately, we can already solve order-preserving pattern matching with $k$ mismatches in time $\bigo(n m \log\log m)$. However, our goal is to develop a faster $\bigo(nf(k))$ time algorithm. For this we cannot afford to verify every possible position using Lemma~\ref{lemma:checking}, and we need a closer look into the structure of the problem.

The first step is to define the {\it signature} of a sequence $(a_{1},\ldots,a_{m})$. Let $\pred(i)$ be the position where the predecessor of $a_i$ among $\{a_1,\ldots,a_m\}$ occurs in the sequence (or $0$, if $a_{i}$ is the smallest element). Then the signature $S(a_{1},\ldots,a_{m})$ is a new sequence $(1-\pred(1),\ldots,m-\pred(m))$ (a simpler version, where the new sequence is $(\pred(1),\ldots,\pred(m))$, was already used to solve the exact version). The signature clearly can be computed in time $\bigo(m\log\log m)$ by sorting. While looking at the signatures is not enough to determine if two sequences are $k$-isomorphic, in some cases it is enough to detect that they are not, as formalized below.

\begin{lemma}
\label{lemma:implication}
If $(a_{1},\ldots,a_{m}) \kmismatches (b_{1},\ldots,b_{m})$, then the Hamming distance between $S(a_{1},\ldots,a_{m})$ and $S(b_{1},\ldots,b_{m})$ is at most $3k$.
\end{lemma}

\begin{proof}
We apply induction on the number of mismatches $k$.  

For $k=0$, $(a_{1},\ldots,a_{m}) \sim (b_{1},\ldots,b_{m})$ iff $S(a_{1},\ldots,a_{m}) = S(b_{1},\ldots,b_{m})$, so the Hamming distance is clearly zero.

Now we proceed to the inductive step. If $(a_{1},\ldots,a_{m}) \mismatches{k+1} (b_{1},\ldots,b_{m})$, then due to Proposition~\ref{prop:inductive}, there exists $(a'_{1},\ldots,a'_{m})$,
such that 
$(a'_1,\ldots,a'_m) \kmismatches (b_{1},\ldots,b_{m})$  and 
$(a_1,\ldots,a_m) \mismatches{1} (a'_{1},\ldots,a'_{m})$. 
Second constraint is equivalent (by application of Lemma~\ref{lemma:subseq}) to existence of such $i$, that $(a_1,\ldots,a_{i-1},a_{i+1},\ldots,a_m) \sim (a'_{1},\ldots,a'_{i-1},a'_{i+1},\ldots,a'_{m}).$

We want to upperbound the Hamming distance between $S(a_1,\ldots,a_m)$ and $S(a'_1,\ldots,a'_m)$. Let $j,j'$ be indices such that $a_j$ is the direct predecessor of $a_i$  and $a_{j'}$ is the direct successor of $a_i$, both taken from the set $\{a_1,\ldots,a_m\}$.
Similarly, let $\ell,\ell'$ be such indices, that $a'_\ell$ is the direct predecessor, and $a'_{\ell'}$ is the direct successor of $a'_i$, both taken from the set $\{a'_1,\ldots,a'_m\}$. That is,
$$\ldots < a_j < a_i < a_{j'} < \ldots $$
is the sorted version of $(a_1,\ldots,a_m),$ and
$$\ldots < a'_{\ell} < a'_i <  a'_{\ell'} < \ldots $$ 
is the sorted version of $(a'_1,\ldots,a'_m)$. The signatures $S(a_1,\ldots,a_m)$ and $S(a'_1,\ldots,a'_m)$ differ on at most 3 positions: $j'$, $\ell'$, and $i$.
Thus $\ham( S(a_1,\ldots,a_m), S(b_1,\ldots,b_m) )$ can be upperbounded by 
$$\ham( S(a_1,\ldots,a_m), S(a'_1,\ldots,a'_m) ) + \ham( S(a'_1,\ldots,a'_m), S(b_1,\ldots,b_m) ) \le 3k + 3,$$
which ends the inductive step.
\qed
\end{proof}

\begin{example}
Consider the following two sequences and their signatures:
\begin{eqnarray*}
S(11,\mathbf{4},12,1,9,3,\mathbf{10},7,2,5,13,0,6,8) &=& (\mathbf{6},\ \ \mathbf{4},-2,\ \ \mathbf{8},9,3,\mathbf{-2},5,-5,\mathbf{-8},\mathbf{-8},0,-3,-6)\\
S(10,\mathbf{1},11,2,9,4,\mathbf{12},7,3,5,13,0,6,8) &=&(\mathbf{4},\mathbf{10},-2,\mathbf{-2},9,3,\mathbf{-4},5,-5,\mathbf{-4},\mathbf{-4},0,-3,-6).
\end{eqnarray*}
One can see easily that the sequences are $2$-isomorphic and the Hamming distance between their signatures is $6$.
\end{example}

Our algorithm iterates through $i=1,2,3,\ldots$ maintaining the signature of the current $(t_{i},\ldots,t_{i+m-1})$. Hence the second step is that we develop in the next two lemmas a data structure, which allows us to store $S(t_{i},\ldots,t_{i+m-1})$, update it efficiently after increasing $i$ by one, and compute its Hamming distance to $S(p_{1},\ldots,p_{m})$.

\begin{lemma}
\label{lemma:prune_structure}
Given a string $S^P[1..m]$, we can maintain a string $S^T[1..2m]$ and perform the following operations:
\begin{enumerate}
\item replacing any character $S^T[x]$ in amortized time $\bigo(\log \log m)$,
\item generating the first $3k$ mismatches between $S^T[i..(i+m-1)]$ and $S^P[1..m]$ in amortized time $\bigo(k+\log \log m)$.
\end{enumerate}
The structure is initialized in time $\bigo(m \log\log m)$.
\end{lemma}

\begin{proof}
We represent the current $S^T[1..2m]$ as a concatenation of a number of fragments. Each fragment is a subword of $S^P$ (possibly single letter) or a special character \$ not occurring in $S^P$. The starting positions of the fragments are kept in a van Emde Boas tree, and additionally each fragment knows its successor and predecessor. In order to bound the amortized complexity of each operation, we maintain an invariant that every element of the tree has $2$ credits available, with one credit being worth $\bigo(\log\log m)$ time. We assume that given any two substrings of $S^P$, we can compute their longest common prefix in $\bigo(1)$ time. This is possible after $\bigo(m)$ preprocessing~\cite{LCA,SuffixArray}.

We initialize the structure by partitioning $S^T$ into $2m$ single characters. The cost of initialization, including allocating the credits, is $\bigo(m \log\log m)$.

Replacing $S^T[x]$ with a new character $c$ starts with locating the fragment $w$ containing the position $i$ using the tree. If $w$ is a single character, we replace it with the new one.
If $w$ is a longer subword $w[i..j]$ of $S^P$, and we need to replace its $\ell$-th character, we first split $w$ into three fragments $w[i..(i+\ell-1)]$, $w[i+\ell]$, $w[(i+\ell+1)..j]$.  In both cases we spend $\bigo(\log \log m)$ time, including the cost of inserting the new elements and allocating their credits.

\begin{figure}[t]
\begin{center}
\includegraphics[width=0.8\linewidth]{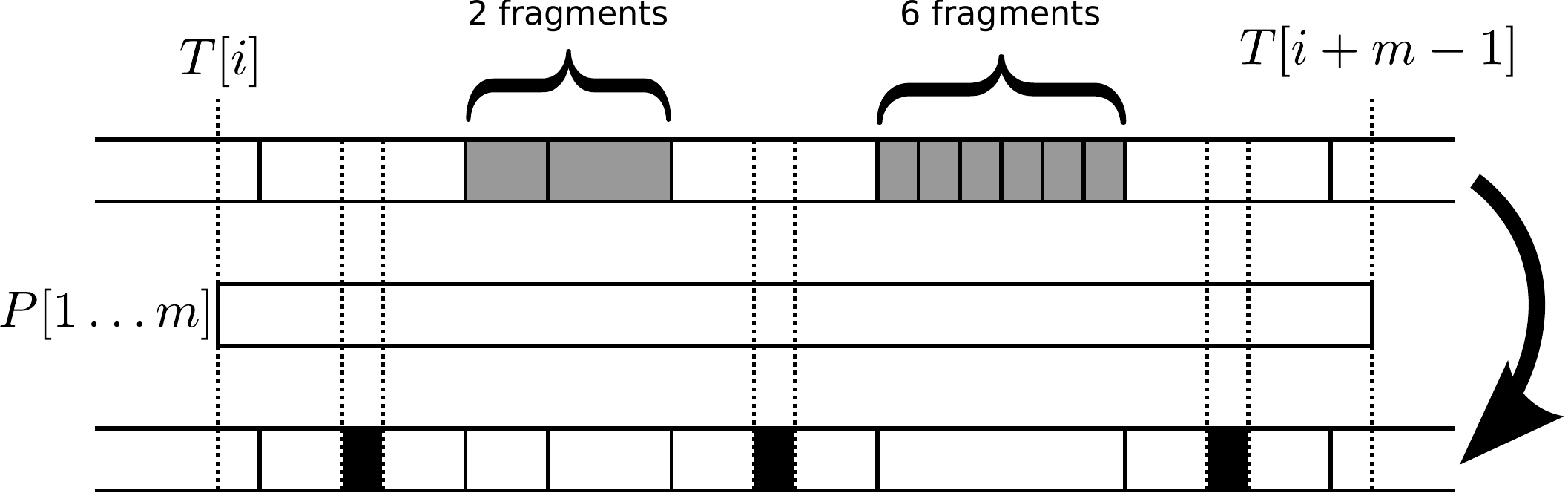}
\end{center}
\caption{Updating the representation. Black boxes represent mismatches, gray areas are full fragments between mismatches. Fragments are either left untouched (on the left), or compressed into a single new one (on the right).}
\label{fig:decomposition}
\end{figure}

Generating the mismatches begins with locating the fragment corresponding to the position $i$. Then we scan the representation from left to right starting from there.
Locating the fragment takes $\bigo(\log \log m)$ time, but traversing can be done in $\bigo(1)$ time per each step, as we can use the information about the successor of each fragment. We will match $S^P$ with the representation of $S^T$ while scanning. This is done using constant time longest common prefix queries. Each such query allows us to either detect a mismatch, or move to the next fragment. Whenever we find a mismatch, if the part of the text between the previous mismatch (or the beginning of the window) and the current mismatch contains at least 3 full fragments, we replace them with a single fragment, which is the corresponding subword of $S^P$. If there are less than $3$ full fragments, we keep the current representation intact, see Fig.~\ref{fig:decomposition}. We stop the scanning after reaching $(3k+1)$-th mismatch, or after the whole window was processed, whichever comes first. 

To bound the amortized cost of processing a single mismatch, let $p$ be the number of full fragments between the current mismatch and the previous one (or the beginning of the window). If $p\ge 3$, we concatenate all $p$ fragments by simply erasing every full fragment except the first one. We also need to traverse (and perform longest common prefix queries on) $p+2$ fragments. However, we remove $p-1$ elements from the tree, and hence can use all their $2(p-1)$ credits to pay for the processing. Thus, the amortized cost is $\bigo((p-1)\log \log m + (p+2) - 2(p-1)\log \log m) = \bigo(1)$. Therefore we need $\bigo(k+\log \log m)$ time in total to generate all the mismatches.
\qed
\end{proof}

\begin{lemma}
\label{lemma:prune}
Given a pattern $(p_{1},\ldots,p_{m})$ and a text $(t_{1},\ldots,t_{2m})$, we can maintain an implicit representation of the current signature $S(t_{i},\ldots,t_{i+m-1})$ and perform the following operations:
\begin{enumerate}
\item increasing $i$ by one in amortized time $\bigo(\log\log m)$,
\item generating the first $3k$ mismatches between $S(p_{1},\ldots,p_{m})$ and $S(t_{i},\ldots,t_{i+m-1})$ in time $\bigo(k+\log\log m)$.
\end{enumerate}
The structure is initialized in time $\bigo(m \log\log m)$.
\end{lemma}

\begin{proof}
First we construct $S(p_{1},\ldots,p_{m})$ in time $\bigo(m\log\log m)$ by sorting. Whenever we increase $i$ by one, just a few characters of $S(t_{i},\ldots,t_{i+m-1})=(s_{1},\ldots,s_{m})$ need to be modified. The new signature can be created by first removing the first character $s_{1}$, appending a new character $s_{m+1}$, and then modifying the characters corresponding to the successors of $t_{i}$ and $t_{i+m}$. By maintaining all $t_{i},\ldots,t_{i+m-1}$ in a van Emde Boas tree (we can rename the elements so that $t_{i} \in \{1,\ldots,2m\}$ by sorting) we can calculate both $s_{m+1}$ and the characters which needs to be modified in $\bigo(\log\log m)$ time. Current $S(t_{i},\ldots,t_{i+m-1})$ is stored using Lemma~\ref{lemma:prune_structure}. We initialize $S^T[1..2m]$ to be $S(t_{1},\ldots,t_{m})$ concatenated with $m$ copies of, say, $0$.  After increasing $i$ by one, we replace $S^T[i]$, $S^T[i+m]$ and possibly two more characters in $\bigo(\log\log m)$ time. Generating the mismatches is straightforward using Lemma~\ref{lemma:prune_structure}.
\qed
\end{proof}

Now our algorithm first uses Lemma~\ref{lemma:prune} to quickly eliminate the starting positions $i$ such that  the Hamming distance between the corresponding signatures is large. For the remaining starting positions, we reduce checking if $(t_{i},\ldots,t_{i+m-1}) \kmismatches (p_{1},\ldots,p_{m})$ to a weighted version of the well-known longest increasing subsequence problem on at most $3(k+1)$ elements. In the weighted variant, which we call \emph{heaviest increasing subsequence}, the input is a sequence $(a_{1},\ldots,a_{\ell})$ and weight $w_{i}$ of each element $a_{i}$, and we look for an increasing subsequence with the largest total weight, i.e., for $1\leq i_{1} < \ldots < i_{s} \leq \ell$ such that $a_{i_{1}} < \ldots < a_{i_{s}}$ and $\sum_{j}w_{i_{j}}$ is maximized.

\begin{lemma}
\label{lemma:reduction}
Assuming random access to $(a_{1},\ldots,a_{m})$, the sorting permutation $\pi_{b}$ of $(b_{1},\ldots,b_{m})$, and the rank of every $b_{i}$ in $\{b_{1},\ldots,b_{m}\}$, and given $\ell$ positions where $S(a_{1},\ldots,a_{m})$ and $S(b_{1},\ldots,b_{m})$ differ, we can reduce in $\bigo(\ell\log\log\ell)$ time checking if $(a_{1},\ldots,a_{m}) \kmismatches (b_{1},\ldots,b_{m})$ to computing the heaviest increasing subsequence on at most $\ell+1$ elements.
\end{lemma}

\begin{figure}[t]
\begin{center}
\includegraphics[width=0.6\linewidth]{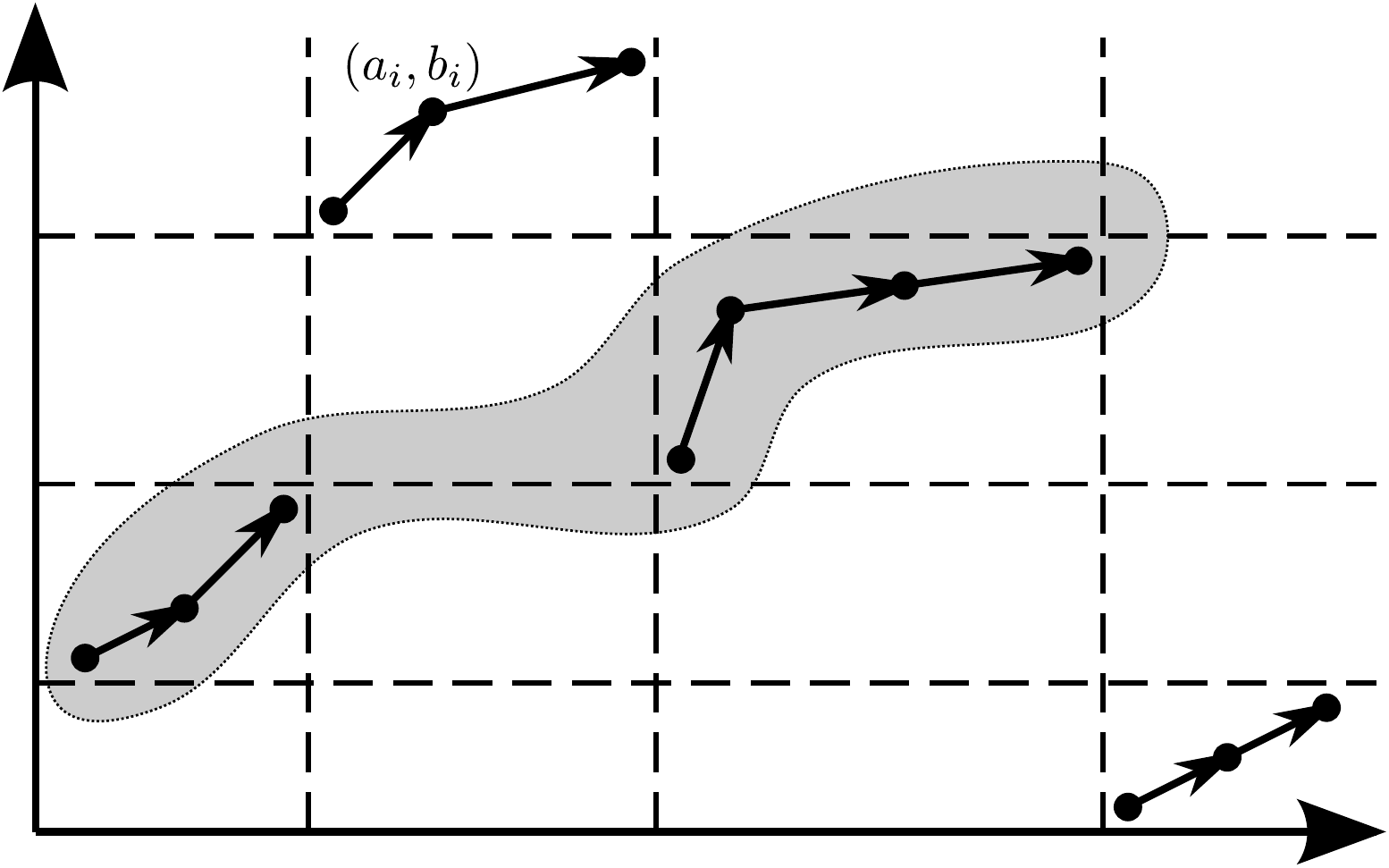}
\end{center}
\caption{Partition into maximal paths. The heaviest increasing subsequence is marked.}
\label{fig:chains}
\end{figure}

\begin{proof}
Let $d_{1},\ldots,d_{\ell}$ be the positions where $S(a_{1},\ldots,a_{m})$ and $S(b_{1},\ldots,b_{m})$ differ. From the definition of a signature, for any other position $i$ the predecessors of $a_{i}$ and $b_{i}$ in their respective sequences are at the same position $j$, which we denote by $j\rightarrow i$.
For any given $i$, $j\rightarrow i$ for at most one $j$. Similarly, for any given $j$, $j\rightarrow i$ for at most one $i$, because the only such $i$ corresponds to the successor of, say, $a_{j}$ in its sequence. Consider a partition of the set of all positions into maximal \emph{paths} of the form $j_{1} \rightarrow \ldots \rightarrow j_{k}$ (see Fig.~\ref{fig:chains}). Such partition is clearly unique, and furthermore the first element of every path is one of the positions where the signatures differ (except one possible path starting with the position corresponding to the smallest element). Hence there are at most $\ell+1$ paths, and we denote by $I_{j}$ the path starting with $d_{j}$. If the smallest element occurs at the same position in both sequences, we additionally denote this position by $d_{0}$, and call the path starting there $I_{0}$ (we will assume that this is always the case, which can be ensured by appending $-\infty$ to both sequences).

Recall that our goal is to check if $(a_{1},\ldots,a_{m}) \kmismatches (b_{1},\ldots,b_{m})$. For this we need to check if  there exist $i_1,\ldots,i_{m-k}$ such that $a_{i_1} < \ldots < a_{i_{m-k}}$ and $b_{i_1} < \ldots < b_{i_{m-k}}$. Alternatively, we could compute the largest $s$ for which there exist a solution $i_1,\ldots,i_{s}$ such that $a_{i_1} < \ldots < a_{i_{s}}$ and $b_{i_1} < \ldots < b_{i_{s}}$. We claim that one can assume that for each path $I$ either none of its elements are among $i_1,\ldots,i_{s}$, or all of its elements are there.  We prove this in two steps.

\begin{enumerate}
\item If $i_{k}\in I$ and $i_{k}\rightarrow j$, then without losing the generality $i_{k+1}=j$. Assume otherwise, so $i_{k+1}\neq j$ or $k=s$. Recall that it means that $a_{j}$ is the successor of $a_{i_{k}}$ and $b_{j}$ is the successor of $b_{i_{k}}$. Hence $a_{i_{k}} < a_{j}$ and $b_{i_{k}} < b_{j}$. If $k=s$ we can extend the current solution by appending $j$. Otherwise $a_{j} < a_{i_{k+1}}$ and $b_{j} < b_{i_{k+1}}$, so we can extend the solution by inserting $j$ between $i_{k}$ and $i_{k+1}$.

\item If $i_{k}\in I$ and $j\rightarrow i_{k}$, then without losing the generality $i_{k-1}=j$. Assume otherwise, so $i_{k-1}\neq j$ or $k=1$. Similarly as in the previous case, $a_{j}$ is the predecessor of $a_{i_{k}}$ and $b_{j}$ is the predecessor of $b_{i_{k}}$. Hence $a_{j} < a_{i_{k}}$ and $b_{j} < b_{i_{k}}$. If $k=1$ we can extend the current solution by prepending $j$. Otherwise $a_{i_{k+1}} < a_{j}$ and $b_{i_{k+1}} < b_{j}$, so we can insert $j$ between $i_{k-1}$ and $i_{k}$.
\end{enumerate}

Now let the weight of a path $I$ be its length $|I|$. From the above reasoning we know that the optimal solution contains either no elements from a path, or all of its elements. Hence if we know which paths contain the elements used in the optimal solution, we can compute $s$ as the sum of the weights of these paths. Additionally, if we take such optimal solution, and remove all but the first element from every path, we get a valid solution. Hence $s$ can be computed by choosing some solution restricted only to $d_{0},\ldots,d_{\ell}$, and then summing up weights of the corresponding paths. It follows that computing the optimal solution can be done, similarly as in the proof of Lemma~\ref{lemma:checking}, by finding an increasing subsequence. We define a new weighted sequence $(a'_{0},\ldots,a'_{\ell})$ by setting $a'_{j}=b_{\pi_{b}(d_{j})}$ and choosing the weight of $a'_{j}$ to be $|I_{j}|$. Then an increasing subsequence of $(a'_{0},\ldots,a'_{\ell})$ corresponds to a valid solution restricted to $d_{0},\ldots,d_{\ell}$, and moreover the weight of the heaviest such subsequence is exactly $s$. In other words, we can reduce our question to computing the heaviest increasing subsequence.

Finally, we need to analyze the complexity of our reduction. Assuming random access to both $(a_{1},\ldots,a_{m})$ and $\pi_{b}$, we can construct $(a'_{0},\ldots,a'_{\ell})$ in time $\bigo(\ell)$. Computing the weight of every $a'_{j}$ is more complicated. We need to find every $|I_{j}|$ without explicitly constructing the paths.  For every $d_{j}$ we can retrieve the rank $r_{j}$ of its corresponding element in $\{b_{1},\ldots,b_{m}\}$. Then $I_{j}$ contains $d_{j}$ and all $i$ such that the predecessor of $b_{i}$ among $\{b_{d_{0}},\ldots,b_{d_{\ell}}\}$ is $b_{d_{j}}$. Hence $|I_{j}|$ can be computed by counting such $i$. This can be done by locating the successor $b_{d_{j'}}$ of $b_{d_{j}}$ in $\{b_{d_{0}},\ldots,b_{d_{\ell}}\}$ and returning $r_{d_{j'}} - r_{d_{j}}-1$ (if the successor does not exist, $m-r_{d_{j}}$). To find all these successors, we only need to sort $\{b_{d_{0}},\ldots,b_{d_{\ell}}\}$, which can, again, be done in time $\bigo(\ell\log\log\ell)$.
\qed
\end{proof}

\begin{lemma}
\label{lemma:heaviest}
Given a sequence of $\ell$ weighted elements, we can compute its heaviest increasing subsequence in time $\bigo(\ell\log\log\ell)$.
\end{lemma}
\begin{proof}
Let the sequence be $(a_{1},\ldots,a_{\ell})$, and denote the weight of $a_{i}$ by $w_{i}$. We will describe how to compute the weight of the heaviest increasing subsequence, reconstructing the subsequence itself will be straightforward. At a high level, for each $i$ we want to compute the weight $r_{i}$ of the heaviest increasing subsequence ending at $a_{i}$. Observe that $r_{i}=w_{i} + \max\{ r_{j} : j<i \text{ and } a_{j} < a_{i}\}$, where we  assume that $a_{0}=-\infty$ and $r_{0}=0$.  We process $i=1,\ldots,\ell$, so we need a dynamic structure where we could store all already computed results $r_{j}$ so that we can select the appropriate one efficiently. To simplify the implementation of this structure, we rename the elements in the sequence so that $a_{i} \in \{1,\ldots,\ell\}$. This can be done in $\bigo(\ell\log\log\ell)$ time by sorting. Then the dynamic structure needs to store $n$ values $v_{1},\ldots,v_{n}$, all initialized to $-\infty$ in the beginning, and implement two operations:

\begin{enumerate}
\item increase any $v_{k}$,
\item given $k$, return the maximum among $v_{1},\ldots,v_{k}$.
\end{enumerate}

Then to compute $r_{i}$ we first find the maximum among $v_{1},\ldots,v_{a_{i}-1}$, and afterwards update $v_{a_{i}}$ to be $r_{i}$.

Now we describe how the structure is implemented. First observe that if $v_{i} > v_{j}$ and $i<j$, we will never return the current $v_{j}$ as the maximum, hence we don't need to store it. In other words, we only need to store $v_{i_{1}}, \ldots, v_{i_{t}}$ such that $i_{1}=1$ and each $i_{j+1}$ is the smallest position on the right of $i_{j}$ such that $v_{i_{j+1}} > v_{i_{j}}$. We store all $i_{j}$ in a van Emde Boas tree, where each element knows its successor and predecessor. Then to find the maximum among $v_{1},\ldots,v_{k}$ we perform a predecessor query in the tree to locate the largest $i_{j}\leq k$, and return the corresponding $v_{i_{j}}$. To increase $v_{k}$, we must consider two cases. First, it might happen that $k=i_{j}$. In such case we update $v_{i_{j}}$ and then look at the successor $i_{j+1}$ of $i_{j}$ in the tree. If $v_{i_{j+1}} \leq v_{i_{j}}$, we remove $i_{j+1}$ and repeat, and otherwise stop. The other case is that $k$ is not in the tree, then we update $v_{k}$ and locate the largest $i_{j}\leq k$. Then if $v_{i_{j}}>v_{k}$, we don't have to modify the tree. Otherwise we need to insert $k$ into the tree, and then repeatedly remove its successors as long as they correspond to elements which are smaller or equal to $v_{k}$, as in the previous case. Finding the maximum clearly requires $\bigo(\log\log\ell)$ time, as it requires just one predecessor query. Increasing any value requires inserting at most one new element, and removing zero or more already existing elements, hence its amortized complexity is $\bigo(\log\log\ell)$. We execute $\ell$ operations in total, so the running time is as claimed.
\qed
\end{proof}


\begin{theorem}
\label{theorem:algo}
Order-preserving pattern matching with $k$ mismatches can be solved in time $\bigo(n(\log\log m+k\log\log k))$, where $n$ is the length of the text and $m$ is the length of the pattern.
\end{theorem}

\begin{proof}
First we focus on the special case when $n\leq 2m$. We iterate over all possible starting positions $i$ in the text while maintaining the signature $S(t_i,\ldots,t_{i+m-1})$ of the current fragment using Lemma~\ref{lemma:prune}. If the Hamming distance between $S(t_i,\ldots,t_{i+m-1})$ and $S(p_1,\ldots,p_m)$ exceeds $3k$, which can be detected in time $\bigo(k+\log\log m)$, by Lemma~\ref{lemma:implication} the current $i$ cannot correspond to a match, and we continue. Otherwise we generate at most $3k$ mismatches,
and apply Lemma~\ref{lemma:reduction} to reduce in time $\bigo(k\log\log k)$ checking if $(t_i,\ldots,t_{i+m-1})\kmismatches (p_1,\ldots,p_m)$ to computing the heaviest increasing subsequence on at most $3(k+1)$ elements. This, by Lemma~\ref{lemma:heaviest}, can be done in time $\bigo(k\log\log k)$, too. We get that the total complexity for a single $i$ is
$\bigo(\log\log m+k\log\log k)$. We spend $\bigo(m\log\log m)$ to initialize the structure from Lemma~\ref{lemma:prune}, so the total time is $\bigo(m(\log\log m+k\log\log k))$.

Finally, by cutting the input into overlapping fragments of length $2m$ and using the above method on each of them, a text of length $n$ can be processed in time $\bigo(\lceil\frac{n}{m}\rceil m(\log\log m+k\log\log k))=\bigo(n(\log\log m+k\log\log k))$.
\qed
\end{proof}

\section{Allowing repeated elements}
\label{section:generalization}

In this section we show how to generalize the algorithm as to remove the restriction that the text (and the pattern) has no repeated elements.
A simple fix could be that instead of comparing numbers, we compare pairs consisting of the number and its position. This might create new occurrences, though.
We will carefully modify all ingredients of the solution described in the previous section to deal with possible equalities. The time complexity will stay the same.

We start with Lemma~\ref{lemma:subseq}. Now the condition becomes that for all $j=1,2,\ldots,m-k-1$, either $a_{i_j} < a_{i_{j+1}}$ and $b_{i_j} < b_{i_{j+1}}$,
or $a_{i_j} = a_{i_{j+1}}$ and $b_{i_j} = b_{i_{j+1}}$. Then checking whether two sequences are $k$-isomorphic seems more complicated, but in fact it is not so. 
We can reduce the question to a generalized heaviest increasing subsequence problem, which we call the \emph{heaviest chain in a plane}.
In this generalization we are given a collection of $\ell$ weighted points in a plane, and the goal is to find the chain with the largest total weight, where a chain is a
set $S$ of points such that for any $(x,y),(x',y')\in S$ either $x<x'$ and $y<y'$, or $x>x'$ and $y>y'$, or $x=x'$ and $y=y'$. This can be solved in $\bigo(\ell\log\log\ell)$ time
similarly as in Lemma~\ref{lemma:heaviest}.

\begin{lemma}
\label{lemma:heaviest chain}
Given $\ell$ weighted points in a plane, we can compute their heaviest chain in time $\bigo(\ell\log\log\ell)$.
\end{lemma}

\begin{proof}
Given $\ell$ weighted points in a plane, we reduce the problem of finding the heaviest chain to computing the heaviest increasing subsequence. First, we make sure that the points
are unique by collapsing all duplicates into single points with the weight equal to the sum of the weights of the collapsed points. Let the $i$-th point be $p_i=(x_i,y_i)$ and define
$x'_i = (x_i, y_i)$ and $y'_i = (y_i, x_i)$.
Even though all $x'_i$ and $y'_i$ are pairs of numbers instead of just numbers, we can still consider the question of computing the heaviest chain for the new set of points 
$p'_i=(x'_i,y'_i)$ if we compare the pairs using the standard lexicographical ordering. Observe that $p'_i$ and $p'_j$ can be in the same chain in the original instance iff $p_i$ 
and $p_j$ can be in the same chain in the new instance, because $x'_i \leq x'_j$ and $y'_i \leq y'_j$ iff $x_i \leq x_j$ and $y_i \leq y_j$. However, in the new instance
we additionally have the property that $x'_i=x'_j$ iff $p_i=p_j$, and symmetrically $y'_i=y'_j$ iff $p_i=p_j$. Since we made sure that the points in the original instance are unique, $x'_i=x'_j$ iff $i=j$, and $y'_i=y'_j$ iff $i=j$. Hence we can reorder the points in the new instance so that their first coordinates are strictly increasing, and then if we look
at the sequence of their second coordinates, its increasing subsequence corresponds to a chain (technically, we also need to normalize the second coordinates by sorting, so that they are numbers instead of pairs of numbers). Hence we can find the heaviest chain in $\bigo(\ell\log\log\ell)$ time using Lemma~\ref{lemma:heaviest}.
\qed
\end{proof}

Now to check if two sequences are $k$-isomorphic, for each $i$ we either create a new point $(a_i,b_i)$ with weight $1$, or if such point already
exists, we increase its weight by $1$. Then we check if the weight of the heaviest chain is at least $n-k$.


We have to modify the definition of the signature. Recall that $S(a_1,\ldots,a_m)$ was defined as a new sequences $(1-\pred(1),\ldots,m-\pred(m))$. Now the predecessor of $a_i$
might be not unique, hence the sequence will consist of pairs from $\{<,=\} \times \mathcal{Z}$. If a given $a_i$ is the rightmost occurrence of the corresponding number,
we output the pair $(<,i-\pred(i))$, where $\pred(i)$ is the position of the predecessor of $(a_i,i)$ in $\{(a_1,1),\ldots,(a_m,m)\}$ (if there is no such
predecessor, $0$), where the pairs are compared using the standard lexicographic ordering. If $a_i$ is not the rightmost occurrence of the corresponding number, we output $(=, i-\pred(i))$.
For such modified definition Lemma~\ref{lemma:implication} still holds. It also holds that $S(t_i,\ldots,t_{i+m-1})$ differs from $S(t_{i+1},\ldots,t_{i+m})$ on a constant
number of positions, hence the time bounds from Lemma~\ref{lemma:prune} remain the same.

Now the only remaining part is to show how to generalize Lemma~\ref{lemma:reduction}. First of all, given that the numbers can repeat, it is not clear what the rank of $b_i$ in
$\{b_1,\ldots,b_m\}$ exactly is. We define it as the number of all elements smaller than $b_i$, and also define \emph{the equal-rank} of $b_i$
to be the number of equal elements on its left. 

\begin{figure*}[t]
\begin{center}
\includegraphics[width=0.5\linewidth]{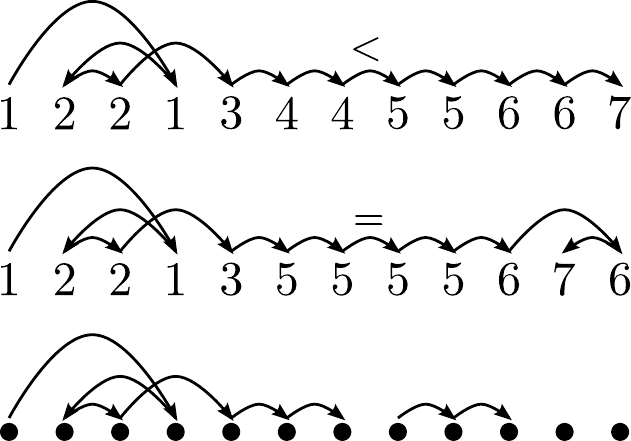}
\hfill
\includegraphics[width=0.4\linewidth]{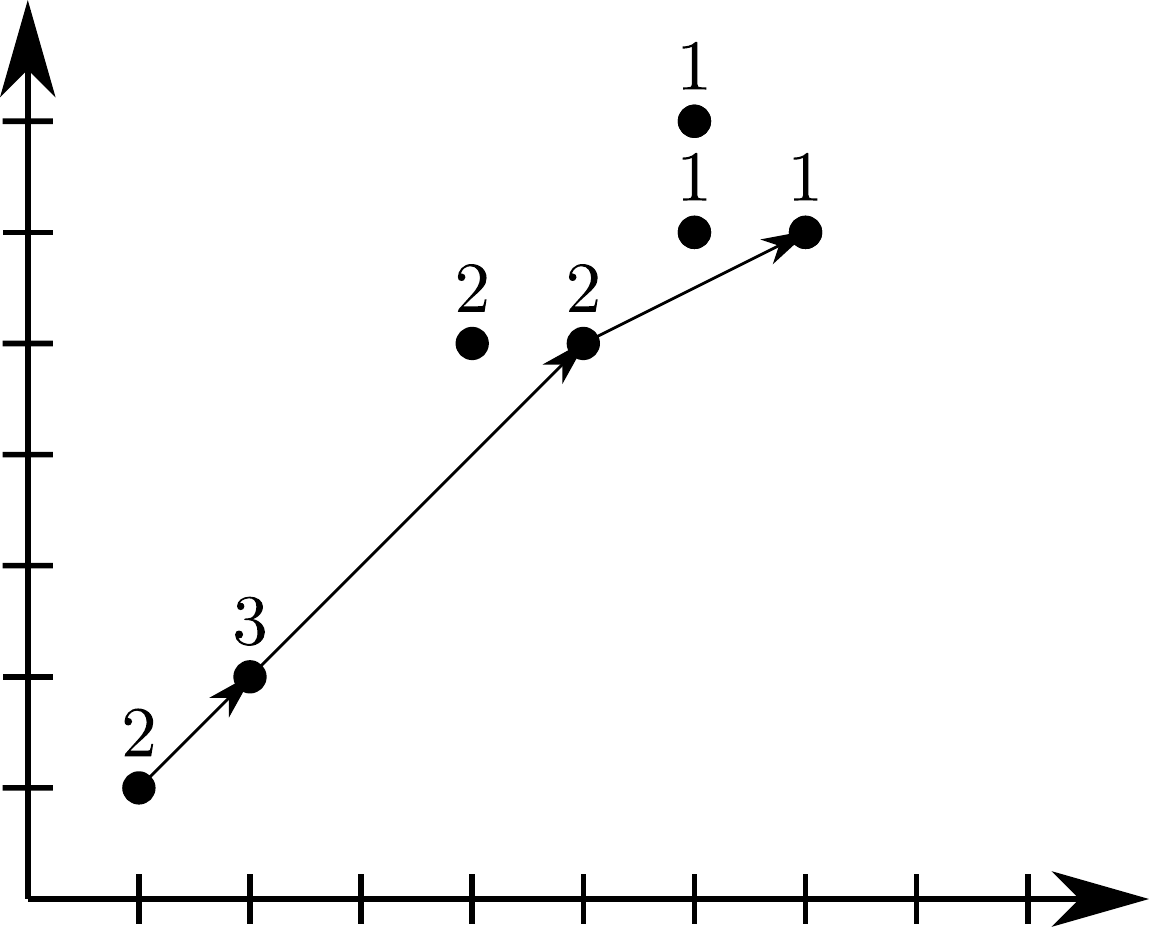}
\end{center}
\caption{Two sequences with their corresponding signatures, which differ on $3$ positions, so there are $4$ maximal paths. The resulting instance of heaviest chain is on the right.
}
\label{fig:eq_paths}
\end{figure*}

\begin{lemma}
\label{lemma:reduction2}
Assuming random access to $(a_{1},\ldots,a_{m})$, the sorting permutation $\pi_{b}$ of $((b_{1},1)\ldots,(b_{m}, m))$, and the rank, the equal-rank,
and total number of repetitions for every $b_i$,
and given $\ell$ positions where $S(a_{1},\ldots,a_{m})$ and $S(b_{1},\ldots,b_{m})$ differ, we can reduce in $\bigo(\ell\log\log\ell)$ time checking if 
$(a_{1},\ldots,a_{m}) \kmismatches (b_{1},\ldots,b_{m})$ to computing the heaviest chain on at most $3(\ell+1)$ elements.
\end{lemma}

\begin{proof}
For any position $i$ where the signatures are the same, either both $a_i$ and $b_i$ are not the rightmost occurrence of the corresponding number, and their next occurrences
are at the same position $j$ in both sequences, denoted $j\stackrel{=}{\rightarrow} i$, or both $a_i$ and $b_i$ are the rightmost occurrence of the corresponding number, and 
the leftmost occurrences of their predecessors are at the same position $j$ in both sequences, denoted $j\stackrel{<}{\rightarrow} i$.
In both cases, we denote the situation by $j\rightarrow i$, and consider the unique partition of the set
of all positions into maximal paths. Now we would like to say that for each such path $i$, either none of its elements belong to the optimal solution, or all of them are
there, where the solution is a collection of indices $i_1,\ldots,i_{m-k}$ such that for all $j=1,2,\ldots,m-k+1$ either $a_{i_j}=a_{i_{j+1}}$ and 
$b_{i_j}=b_{i_{j+1}}$, or $a_{i_j}<a_{i_{j+1}}$ and $b_{i_j}<b_{i_{j+1}}$.
Unfortunately, this is not true: one path might end at some $i$, and the other might start at some $j$, such that $a_i=a_j$, yet $b_i\neq b_j$. Then we cannot have
these two whole paths in the solution, but it might pay off to have a prefix of the former, or a suffix of the latter, see Fig.~\ref{fig:eq_paths}.
Our fix is to additionally split every path into three parts. The parts correspond to the maximal prefix $I_{pref}$ of the form $i_1\stackrel{=}{\rightarrow}i_2\stackrel{=}{\rightarrow}\ldots$,
the middle part $I_{middle}$, and the maximal suffix $I_{suf}$ of the form $\ldots\stackrel{=}{\rightarrow}i_{k-1}\stackrel{=}{\rightarrow}i_k$
The splitting can be performed efficiently using the ranks, the equal-ranks, and the total number of repetitions. Then we create an instance
of the heaviest chain problem by collapsing each path into a single weighted point (with the same coordinates as the first point on the path), and additionally replacing identical points with one (and summing up their weights).

Now we need to prove that solving this instance gives us an optimal solution to the original question.
Clearly, if $j\stackrel{=}{\rightarrow} i$ then the optimal solution takes both $i$ and $j$ or none of them, hence merging identical points preserves the optimal solution. We must show 
that for each chain $I$ decomposed into $I_{pref}\cup I_{middle}\cup I_{suf}$ the optimal solution contains either all points from $I_{middle}$ or none of them.
Let $I_{pref} = \ldots \stackrel{=}{\rightarrow}i'$, $I_{middle}=i \rightarrow \ldots \rightarrow j$, and $I_{suf} = j'\stackrel{=}{\rightarrow}\ldots$. 
We know that $a_{i'}<a_{i}$ and $b_{i'}<b_{i}$, and also $a_{j}<a_{j'}$ and $b_{j}<b_{j'}$. Hence for all $k\in I_{middle}$, all positions $k'$ such that $a_k=a_{k'}$ belong
to $I_{middle}$, and similarly all positions $k'$ such that $b_k=b_{k'}$ belong to $I_{middle}$. Furthermore, these two sets of positions are the same.
It follows that after collapsing identical points we get that for every $k\in I_{middle}$, there are no $k'$ (inside or outside $I_{middle}$) such that $a_k=a_{k'}$ or
$b_k=b_{k'}$. A straightforward modification of the two step proof used in Lemma~\ref{lemma:reduction} can be used to show that if one element from $I_{middle}$ belongs to the
optimal solution, then all its elements are there.
\qed
\end{proof}

\section{Conclusions}

Recall that the complexity of our solution is $\bigo(n(\log\log m+k\log\log k))$.
Given that it is straightforward to prove a lower bound of
$\Omega(n+m\log m)$ in the comparison model, and that for $k=0$ one can achieve $\bigo(n+\text{sort}(m))$ time~\cite{KubicaOrder}, a natural question is whether achieving $\bigo(nf(k))+\bigo(m\polylogws(m))$ time is possible.
Finally, even though the version with $k$ errors seems hard (see the introduction), there might be an $\bigo(nf(k))$ time algorithm, with $f(k)$ being an exponential function.

\bibliographystyle{splncs03}
\bibliography{biblio}


\end{document}